\documentclass[11pt]{article}

\usepackage[margin=1in]{geometry}
\usepackage{amsthm}

\newtheorem{corollary}{Corollary}
\newtheorem{lemma}{Lemma}
\newtheorem{proposition}{Proposition}
\newtheorem{theorem}{Theorem}

\theoremstyle{remark}
\newtheorem{remark}{Remark}

\begin{document}

\title{Computing quantum discord is NP-complete}

\author{Yichen Huang\\Department of Physics, University of California, Berkeley, CA 94720, USA}

\date{\today}

\maketitle

\begin{abstract}

We study the computational complexity of quantum discord (a measure of quantum correlation beyond entanglement), and prove that computing quantum discord is NP-complete. Therefore, quantum discord is computationally intractable: the running time of any algorithm for computing quantum discord is believed to grow exponentially with the dimension of the Hilbert space so that computing quantum discord in a quantum system of moderate size is not possible in practice. As by-products, some entanglement measures (namely entanglement cost, entanglement of formation, relative entropy of entanglement, squashed entanglement, classical squashed entanglement, conditional entanglement of mutual information, and broadcast regularization of mutual information) and constrained Holevo capacity are NP-hard/NP-complete to compute. These complexity-theoretic results are directly applicable in common randomness distillation, quantum state merging, entanglement distillation, superdense coding, and quantum teleportation; they may offer significant insights into quantum information processing. Moreover, we prove the NP-completeness of two typical problems: linear optimization over classical states and detecting classical states in a convex set, providing evidence that working with classical states is generically computationally intractable.

\end{abstract}

\section{Introduction}

Quite a few fundamental concepts in quantum mechanics do not have classical analogs: uncertainty relations \cite{Bec75, BM75, Hua11, Hua12, Rob29}, quantum nonlocality \cite{CHSH69, EPR35, HHHH09, PV07}, etc. Quantum entanglement \cite{HHHH09, PV07}, defined based on the notion of local operations and classical communication (LOCC), is the most prominent manifestation of quantum correlation. It is a resource in quantum information processing, enabling tasks such as superdense coding \cite{BW92}, quantum teleportation \cite{BBC+93} and quantum state merging \cite{HOW05, HOW07}. Various entanglement measures \cite{HHHH09, PV07} are reported to quantify entanglement. However, nontrivial quantum correlation also exists in certain separable (not entangled) states. For instance, deterministic quantum computation with one qubit (DQC1) \cite{KL98} is a model of mixed-state quantum computation with little entanglement. It is argued \cite{DSC08} that quantum discord \cite{HV01, OZ01} (a measure of quantum correlation beyond entanglement; see section \ref{QD} for its definition) is responsible for the quantum speed-up over classical algorithms. Quantum discord is also a useful concept in common randomness distillation \cite{DW04}, quantum state merging \cite{CAB+11, MD11, MD13}, entanglement distillation \cite{MD13, SKB11}, superdense coding \cite{MD13}, quantum teleportation \cite{MD13}, etc, and has established quantum discord (and related measures of quantum correlation) as an active research topic over the past few years \cite{MBC+12}. Nevertheless, computing quantum discord is difficult. Despite considerable effort, few analytical results are known even for `simple' and useful states (e.g. two-qubit $X$ states \cite{ARA10, ARA10E, CZY+11, Dil08, Hua13a, Hua14, LMXW11, Sar09}). Generally, quantum discord can only be computed numerically.

The notion of NP-completeness \cite{Coo71} is fundamental and remarkable in computational complexity theory. NP-complete problems are the hardest in NP in the sense that an efficient algorithm for any NP-complete problem implies efficient algorithms for all problems in NP, and NP-hard problems are at least as hard as NP-complete problems. An NP-hard/NP-complete problem is computationally intractable: the running time of any algorithm for the problem is believed to grow exponentially with the input size. The NP-completeness of the separability problem (detecting whether a given state is separable) was first proved in \cite{Gur03, Gur04}; see \cite{Gha10, Ioa07} for technical improvements. This may be the reason why a lot of effort is devoted to entanglement criteria \cite{DGCZ00, GT09, Hor97, HHHH09, Hua10, Hua10E, Hua13, Per96, Sim00}, which are simple sufficient conditions for entanglement. The classicality problem (detecting whether a given state has zero quantum discord) can be solved in polynomial time \cite{CCMV11, DVB10, HWZ11}, but the computational complexity of quantum discord is not known.

Here we prove that computing quantum discord is NP-complete (theorem \ref{main}). Therefore, the running time of any algorithm for computing quantum discord is believed to grow exponentially with the dimension of the Hilbert space, so that computing quantum discord in a quantum system of moderate size is not possible in practice. As by-products, some entanglement measures (namely entanglement cost \cite{BDSW96}, entanglement of formation \cite{BDSW96}, relative entropy of entanglement \cite{VPRK97}, squashed entanglement \cite{CW04}, classical squashed entanglement, conditional entanglement of mutual information \cite{YHW08}, and broadcast regularization of mutual information \cite{PCMH09}; theorem \ref{EM}) and constrained Holevo capacity \cite{Sho04} (corollary \ref{cHolevo}) are NP-hard/NP-complete to compute. As direct applications (one-way), distillable common randomness, regularized one-way classical deficit, entanglement consumption in extended quantum state merging, and minimum loss due to decoherence of the yield of a family of protocols are also NP-hard/NP-complete to compute; such complexity-theoretic results may offer significant insights into quantum information processing. Moreover, we prove the NP-completeness of the following two typical problems: linear optimization over classical states (proposition \ref{prop1}) and detecting whether there are classical states in a given convex set (proposition \ref{prop2}). The former is the simplest optimization problem over classical states, and the latter is just one step further than the classicality problem. Conceptually, the NP-completeness of these two problems provides evidence that working with classical states is generically computationally intractable. We conclude with some interesting open problems and research directions.

\section{NP-hardness/NP-completeness of computing entanglement measures}

Let us briefly recall the definitions of some entanglement measures (see the review papers \cite{HHHH09, PV07} for details). Entanglement cost $E_C(\rho)$ \cite{BDSW96} is the minimum rate $j/k$ to convert $j$ copies of the two-qubit maximally entangled state $|\psi\rangle=(|00\rangle+|11\rangle)/\sqrt2$ to $k$ copies of the bipartite state $\rho$ by LOCC with vanishing error in the asymptotic limit $j,k\rightarrow+\infty$. Conversely, distillable entanglement $E_D(\rho)$ \cite{BDSW96} is the maximum rate $j/k$ to convert $\rho^{\otimes k}$ to $|\psi\rangle^{\otimes j}$ by LOCC with vanishing error in the asymptotic limit. Entanglement of formation \cite{BDSW96} is defined as
\begin{equation}
E_F(\rho_{AB})=\inf_{\{p_i,|\psi_i\rangle\}}\sum_ip_iS(\rho_A^i),
\label{Ef}
\end{equation}
where the infimum is taken over all ensembles of pure states $\{p_i,|\psi_i\rangle\}$ satisfying $\rho_{AB}=\sum_ip_i|\psi_i\rangle\langle\psi_i|$, and
\begin{equation}
S(\rho_A^i)=-\mathrm{tr}(\rho_A^i\log\rho_A^i)
\end{equation}
is the von Neumann entropy of the reduced density matrix $\rho_A^i=\mathrm{tr}_B|\psi_i\rangle\langle\psi_i|$ or the entanglement entropy of $|\psi_i\rangle$. The relative entropy of entanglement \cite{VPRK97}
\begin{equation}
E_R(\rho)=\inf_{\sigma\in\cal{S}}S(\rho\|\sigma)=\inf_{\sigma\in\cal{S}}\mathrm{tr}(\rho\log\rho-\rho\log\sigma)
\end{equation}
quantifies the distance from the state $\rho$ to the set $\cal{S}$ of all separable states, where $S(\rho\|\sigma)$ is the quantum relative entropy. The regularized relative entropy of entanglement is given by
\begin{equation}
E_R^\infty(\rho)=\lim_{k\rightarrow+\infty}E_R(\rho^{\otimes k})/k.
\end{equation}
Squashed entanglement \cite{CW04} is defined as
\begin{equation}
E_{sq}(\rho_{AB})=\frac{1}{2}\inf_{\rho_{ABC}}\{S(\rho_{AC})+S(\rho_{BC})-S(\rho_C)-S(\rho_{ABC})\},
\label{sq}
\end{equation}
where the infimum is taken over all states $\rho_{ABC}$ in an extended Hilbert space satisfying $\rho_{AB}=\mathrm{tr}_C\rho_{ABC}$, and $\rho_{AC}=\mathrm{tr}_B\rho_{ABC}$, etc. Classical squashed entanglement $E_{sq}^C(\rho_{AB})$ is given by (\ref{sq}) where the infimum is taken with the additional restriction that $\rho_{AB|C}$ is quantum-classical (\ref{QC}) across the cut $AB|C$. Conditional entanglement of mutual information \cite{YHW08} is defined as
\begin{equation}
E_I(\rho_{AB})=\frac{1}{2}\inf_{\rho_{AA'BB'}}\{I(\rho_{AA'|BB'})-I(\rho_{A'B'})\},
\end{equation}
where the infimum is taken over all states $\rho_{AA'BB'}$ satisfying $\rho_{AB}=\mathrm{tr}_{A'B'}\rho_{AA'BB'}$, and
\begin{equation}
I(\rho_{A'B'})=S(\rho_{A'})+S(\rho_{B'})-S(\rho_{A'B'})
\end{equation}
is the quantum mutual information. A state $\rho_{X^{\otimes k}}$ with $X^{\otimes k}=\otimes_{i=1}^kX_i$ is a $k$-copy broadcast state of $\rho_X$ if $\rho_X=\mathrm{tr}_{X_1X_2\cdots X_{i-1}X_{i+1}\cdots X_k}\rho_{X^{\otimes k}}$ for any $i=1,2,\ldots,k$. Broadcast regularization of mutual information \cite{PCMH09} is given by
\begin{equation}
I_b^\infty(\rho_{AB})=\frac{1}{2}\lim_{k\rightarrow+\infty}\inf_{\rho_{A^{\otimes k}|B^{\otimes k}}}I(\rho_{A^{\otimes k}|B^{\otimes k}})/k,
\end{equation}
where the infimum is taken over all $k$-copy broadcast states of $\rho_{AB}$.

\begin{lemma}
(a) The definition (\ref{Ef}) of $E_F$ remains the same if the number of states in the ensemble is restricted to be less than or equal to $m^2n^2$ \cite{Nie00, Uhl98}, where $m\times n$ is the dimension of the bipartite state $\rho_{AB}$.

(b) $E_C(\rho)=E_F^\infty(\rho)=\lim_{k\rightarrow+\infty}E_F(\rho^{\otimes k})/k$ \cite{HHT01}.

(c) $E_F(\rho)\ge E_R(\rho)$ \cite{VP98}, $E_C(\rho)\ge E_{sq}(\rho)$ \cite{CW04}, $E_{sq}^C(\rho)\ge I_b^\infty(\rho)\ge E_I(\rho)\ge E_{sq}(\rho)$ \cite{PCMH09}.

(d) $E_R^\infty(\rho)\ge\inf_{\sigma\in\cal{S}}\|\rho-\sigma\|_1^2/(2mn\log2)$ \cite{Pia09}, $E_{sq}(\rho)\ge\inf_{\sigma\in\cal{S}}\|\rho-\sigma\|_2^2/(2448\log2)$ \cite{BCY11, BCY12, MWW09}, where $\|X\|_1=\mathrm{tr}\sqrt{X^\dag X}$ and $\|X\|_2=\sqrt{\mathrm{tr}X^\dag X}$ are the trace norm and the Frobenius norm, respectively.
\label{pre}
\end{lemma}

Accounting for the finite precision of numerical computing, hereafter, every real number is assumed to be represented by a polynomial number of bits, and the formulation of each computational problem is approximate. Indeed, we will prove that the problems are computationally intractable even if small errors are allowed. We begin by recalling the following lemma.

\begin{lemma}[NP-completeness of the separability problem]
Given a bipartite quantum state $\rho$ of dimension $m\times n$ with the promise that either (Y) $\rho\in\cal{S}$ or (N) $\inf_{\sigma\in\cal{S}}\|\rho-\sigma\|_2\ge\delta$, it is NP-complete to decide which is the case, where $\delta=1/\mathrm{poly}(m,n)$ is some inverse polynomial in $m,n$.
\label{Sep}
\end{lemma}

\begin{remark}
The NP-completeness of the separability problem with $\delta=\exp(-O(m,n))$ is proven in \cite{Ioa07}, and the NP-hardness of the separability problem with $\delta=1/\mathrm{poly}(m,n)$ is proven in \cite{Gha10}. The separability problem can be solved in $\exp(O((\log m)(\log n)/\delta^2))$ time (a quasi-polynomial-time algorithm for $\delta=1/\mathrm{poly}(\log m,\log n)$) \cite{BCY11a}.
\end{remark}

\begin{theorem}[NP-hardness/NP-completeness of computing entanglement measures]
Given a bipartite quantum state $\rho$ of dimension $m\times n$ and a real number $a$ with the promise that either (Y) $E_F(\rho)\le a$ or (N) $E_F(\rho)\ge a+\epsilon$, it is NP-complete to decide which is the case, where $\epsilon=1/\mathrm{poly}(m,n)$. In the same sense, computing $E_R$ is NP-complete and computing $E_C, E_R^\infty, E_{sq}, E_{sq}^C, E_I, I_b^\infty$ is NP-hard.
\label{EM}
\end{theorem}

\begin{proof}
Computing $E_F,E_R$ is in NP: the certificates of the yes instances (Y) are the optimal ensemble of pure states $\{p_i,|\psi_i\rangle\}$ and the closest separable state $\sigma$, respectively. The NP-hardness of computing entanglement measures is totally expected, as computing entanglement measures is more difficult than just detecting entanglement. Indeed, the hardness proof is a reduction from lemma \ref{Sep}. Set $a=0$ and $\epsilon=\delta^2/(2448mn\log2)=1/\mathrm{poly}(m,n)$. (Y) If $\rho$ is separable, then
\begin{equation}
E_C(\rho)=E_F(\rho)=E_R(\rho)=E_R^\infty(\rho)=E_{sq}(\rho)=E_{sq}^C(\rho)=E_I(\rho)=I_b^\infty(\rho)=0.
\end{equation}
(N) If $\inf_{\sigma\in\cal{S}}\|\rho-\sigma\|_2\ge\delta$, then
\begin{eqnarray}
&&E_F(\rho)\ge E_R(\rho)\ge E_R^\infty(\rho),~E_F(\rho)\ge E_C(\rho)\ge E_{sq}(\rho),~E_{sq}^C(\rho)\ge I_b^\infty(\rho)\ge E_I(\rho)\ge E_{sq}(\rho),\\
&&E_R^\infty(\rho)\ge\inf_{\sigma\in\cal{S}}\|\rho-\sigma\|_1^2/(2mn\log2)\ge\inf_{\sigma\in\cal{S}}\|\rho-\sigma\|_2^2/(2mn\log2)\ge\delta^2/(2mn\log2)\ge\epsilon,\\
&&E_{sq}(\rho)\ge\inf_{\sigma\in\cal{S}}\|\rho-\sigma\|_2^2/(2448\log2)\ge\delta^2/(2448\log2)\ge\epsilon.
\end{eqnarray}
\end{proof}

\begin{remark}
The computational problem in theorem \ref{EM} requires a guess of $E_F(\rho)$ as an input. This formulation is reasonable: if there is an efficient subroutine for the problem, a binary search for $E_F(\rho)$ can be done by calling the subroutine $O(\log(\log(mn)/\epsilon))=O(\log m,\log n)$ times to achieve the precision $\epsilon=1/\mathrm{poly}(m,n)$. The hardness proof does not apply to $E_D$, as $E_D(\rho)$ can be zero for an entangled state $\rho$. It is an open problem whether computing $E_C, E_R^\infty, E_{sq}, E_{sq}^C, E_I, I_b^\infty$ is in NP. For instance, it is not clear how large the dimension of $\rho_{ABC}$ should be so that the right-hand side of (\ref{sq}) is optimal (or sufficiently close to optimal).

\end{remark}

\section{NP-completeness of computing quantum discord}
\label{QD}

As a measure of quantum correlation (beyond entanglement), quantum discord \cite{OZ01}
\begin{equation}
D(\rho_{AB}|B)=I(\rho_{AB})-J(\rho_{AB}|B)
\end{equation}
is the difference between total correlation (quantified by quantum mutual information) and classical correlation \cite{HV01}
\begin{equation}
J(\rho_{AB}|B)=S(\rho_A)-\inf_{\{\Pi_i\}}\sum_ip_iS(\rho_A^i),
\end{equation}
where $\{\Pi_i\}$ is a measurement on the subsystem $B$; $p_i=\mathrm{tr}(\rho_{AB}\Pi_i)$ is the probability of the $i$th measurement outcome; $\rho_A^i=\mathrm{tr}_B(\rho_{AB}\Pi_i)/p_i$ is the post-measurement state. The infimum is taken over either all von Neumann measurements or all generalized measurements described by positive-operator valued measures (POVM); the corresponding notations are $J_N,D_N$ and $J_P,D_P$, respectively. See \cite{NC11} for an introduction to von Neumann measurements and POVM measurements. The definitions of $J_P,D_P$ remain the same if the number of operators in the POVM is restricted to be less than or equal to $n^2$, where $n$ is the dimension of the subsystem $B$. This is because the optimal POVM must be extremal \cite{HKZ04}, and an extremal POVM contains at most $n^2$ operators \cite{DPP05}. Regularized classical correlation and quantum discord are given by
\begin{equation}
J^\infty(\rho_{AB}|B)=\lim_{k\rightarrow+\infty}J(\rho_{AB}^{\otimes k}|B^{\otimes k})/k,~D^\infty(\rho_{AB}|B)=I(\rho_{AB})-J^\infty(\rho_{AB}|B)=\lim_{k\rightarrow+\infty}D(\rho_{AB}^{\otimes k}|B^{\otimes k})/k.
\end{equation}

\begin{theorem}[NP-completeness of computing quantum discord]
Given a bipartite quantum state $\rho_{AB}$ of dimension $m\times n$ and a real number $b$ with the promise that either (Y) $D_P(\rho_{AB}|B)\le b$ or (N) $D_P(\rho_{AB}|B)\ge b+\epsilon$, it is NP-complete to decide which is the case, where $\epsilon=1/\mathrm{poly}(m,n)$. In the same sense, computing $D_N,J_{N,P}$ is NP-complete and computing $D_{N,P}^\infty,J_{N,P}^\infty$ is NP-hard.
\label{main}
\end{theorem}

\begin{proof}
Computing $D_{N,P}$ is in NP: the certificates of (Y) are the optimal measurements $\{\Pi_i\}$ on the subsystem $B$. The hardness proof is basically a reduction from theorem \ref{EM} via the Koashi-Winter relation \cite{KW04} between $E_F$ and $D_P$, and technically we derive a similar relation between $E_F$ and $D_N$ (note that the Koashi-Winter relation is between $E_F$ and $D_P$ rather than between $E_F$ and $D_N$). Given a bipartite state $\rho_{AB}$ of dimension $m\times n$, by diagonalizing $\rho_{AB}$ we construct a tripartite pure state $|\Psi_{ABC}\rangle$ of dimension $m\times n\times m^2n^2$ satisfying $\rho_{AB}=\mathrm{tr}_C|\Psi_{ABC}\rangle\langle\Psi_{ABC}|$ (note that such a tripartite pure state of dimension $m\times n\times mn$ exists, but a larger dimension of the subsystem $C$ will be useful later). (i) A POVM measurement $\{\Pi_i\}$ on $C$ produces an ensemble $\{p_i,\rho_i\}$ satisfying $\rho_{AB}=\sum_ip_i\rho_i$, where $p_i=\mathrm{tr}(|\Psi_{ABC}\rangle\langle\Psi_{ABC}|\Pi_i)$ and $\rho_i=\mathrm{tr}_{C}(|\Psi_{ABC}\rangle\langle\Psi_{ABC}|\Pi_i)/p_i$. (ii) For any ensemble $\{p_i,\rho_i\}$ satisfying $\rho_{AB}=\sum_ip_i\rho_i$, a POVM measurement $\{\Pi_i\}$ exists on $C$ such that $p_i=\mathrm{tr}(|\Psi_{ABC}\rangle\langle\Psi_{ABC}|\Pi_i)$ and $\rho_i=\mathrm{tr}_C(|\Psi_{ABC}\rangle\langle\Psi_{ABC}|\Pi_i)/p_i$; moreover, such a von Neumann measurement on $C$ exists if the dimension of $C$ is greater than or equal to the number of states in the ensemble \cite{HJW93} (this condition is satisfied due to lemma \ref{pre}(a)). As the definition (\ref{Ef}) of $E_F$ remains the same if the infimum is taken over all ensembles of possibly mixed states $\{p_i,\rho_i\}$ satisfying $\rho_{AB}=\sum_ip_i\rho_i$, the relation
\begin{equation}
E_F(\rho_{AB})=D_{N,P}(\rho_{BC}|C)+S(\rho_A)-S(\rho_{AB})
\end{equation}
follows immediately from the definitions of $E_F$ and $D_{N,P}$ (note that in the present case $D_N=D_P$, though generically $D_N\neq D_P$). Set $b=a-S(\rho_A)+S(\rho_{AB})$. We complete the reduction from $E_F$ to $D_{N,P}$ by taking $\rho_{BC}$ as the input to the computational problem in theorem \ref{main}. The regularized relation
\begin{equation}
E_C(\rho_{AB})=D_{N,P}^\infty(\rho_{BC}|C)+S(\rho_A)-S(\rho_{AB}).
\end{equation}
implies a reduction from $E_C$ to $D_{N,P}^\infty$. These reductions are polynomial-time reductions.
\end{proof}

\section{NP-completeness of computing constrained Holevo capacity}

A quantum channel $\Phi$ is a completely positive trace-preserving linear map \cite{NC11} from states of dimension $n_i$ to states of dimension $n_o$. The constrained Holevo capacity \cite{Sho04} is defined as
\begin{equation}
\chi_\Phi(\rho)=S(\Phi(\rho))-\inf_{\{p_i,|\psi_i\rangle\}}\sum_ip_iS(\Phi(|\psi_i\rangle\langle\psi_i|)),
\label{constrain}
\end{equation}
where the infimum is taken over all ensembles of pure states $\{p_i,|\psi_i\rangle\}$ satisfying $\rho=\sum_ip_i|\psi_i\rangle\langle\psi_i|$. The definition (\ref{constrain}) of $\chi_\Phi(\rho)$ remains the same if the number of states in the ensemble is restricted to be less than or equal to $n_i^2$ \cite{NC11}. The regularized constrained Holevo capacity is given by
\begin{equation}
\chi_\Phi^\infty(\rho)=\lim_{k\rightarrow+\infty}\chi_{\Phi^{\otimes k}}(\rho^{\otimes k})/k.
\end{equation}
The Holevo capacity $\chi_\Phi=\sup_{\rho}\chi_\Phi(\rho)$ is the maximum rate at which classical information can be transmitted through the quantum channel $\Phi$ using product states as codewords \cite{Hol98, SW97}. The regularized Holevo capacity is given by
\begin{equation}
\chi_\Phi^\infty=\lim_{k\rightarrow+\infty}\chi_{\Phi^{\otimes k}}/k\neq\sup_\rho\chi_\Phi^\infty(\rho).
\end{equation}

\begin{corollary}[NP-completeness of computing constrained Holevo capacity]
Given a quantum channel $\Phi$, a quantum state $\rho$ of dimension $n_i$, and a real number $c$ with the promise that either (Y) $\chi_\Phi(\rho)\ge c$ or (N) $\chi_\Phi(\rho)\le c-\epsilon$, it is NP-complete to decide which is the case, where $\epsilon=1/\mathrm{poly}(n_i,n_o)$. In the same sense, computing $\chi_\Phi^\infty(\rho)$ is NP-hard.
\label{cHolevo}
\end{corollary}

\begin{proof}
Computing $\chi_{\Phi}(\rho)$ is in NP: the certificate of (Y) is the optimal ensemble of pure states $\{p_i,|\psi_i\rangle\}$. The hardness proof is a reduction from theorem \ref{EM} via the relation \cite{MSW04, Sho04} between $E_F$ and $\chi_\Phi(\rho)$. Given a bipartite state $\sigma_{AB}$ of dimension $m\times n$, let $U$ be a unitary embedding such that $\sigma_{AB}=U(\rho)$ for a state $\rho$ of dimension $\mathrm{rank}(\sigma_{AB})$. The quantum channel $\Phi$ is defined as $\Phi(\rho')=\mathrm{tr}_BU(\rho')$, where $n_i=\mathrm{rank}(\sigma_{AB})=O(mn)$ and $n_o=m$. Then
\begin{equation}
E_F(\sigma_{AB})=S(\Phi(\rho))-\chi_\Phi(\rho).
\end{equation}
Set $c=S(\Phi(\rho))-a$. We complete the reduction from $E_F(\sigma_{AB})$ to $\chi_\Phi(\rho)$. The regularized relation
\begin{equation}
E_C(\sigma_{AB})=S(\Phi(\rho))-\chi_\Phi^\infty(\rho)
\end{equation}
implies a reduction from $E_C(\sigma_{AB})$ to $\chi_\Phi^\infty(\rho)$. These reductions are polynomial-time reductions.
\end{proof}

\begin{lemma}[NP-completeness of computing Holevo capacity]
Given a quantum channel $\Phi$ and a real number $c$ with the promise that either (Y) $\chi_\Phi\ge c$ or (N) $\chi_\Phi\le c-\epsilon$, it is NP-complete to decide which is the case, where $\epsilon=1/\mathrm{poly}(n_i,n_o)$.
\label{Holevo}
\end{lemma}

\begin{remark}
This is one of the main results of \cite{BS07}, in which, however, the scaling of $\epsilon$ is not discussed. Indeed, additional work is needed to establish the NP-completeness of computing $\chi_\Phi$ with $\epsilon=1/\mathrm{poly}(n_i,n_o)$. We are not going to present the complete proof here. The computational complexity of $\chi_\Phi^\infty$ remains an open problem. The set of all states of dimension $n_i$ is convex, and $-\chi_\Phi(\rho)$ is a convex function as $S(\Phi(\rho))$ is concave and the infimum in (\ref{constrain}) is convex. Thus an alternative proof of the NP-hardness of computing $\chi_\Phi(\rho)$ is a polynomial-time reduction from lemma \ref{Holevo} via convex optimization \cite{BV04}; moreover, the NP-hardness of computing $E_F$ can be proved\footnote{Mark M Wilde, private communication.} based on lemma \ref{Holevo}.
\end{remark}

\section{Applications}

Common randomness is a resource in information theory and cryptography \cite{AC93, AC98}. One-way distillable common randomness $D_{cr}(\rho_{AB}|B)$ is the maximum rate at which common randomness can be extracted from the bipartite state $\rho_{AB}$ by local operations and one-way classical communication in the asymptotic limit. It is equal to regularized classical correlation $J_P^\infty(\rho_{AB}|B)$ \cite{DW04}, and also equal to regularized one-way classical deficit \cite{DW05}. Thus $D_{cr}$ and regularized one-way classical deficit are NP-hard to compute.

In quantum state merging, Alice and Bob share a bipartite state, and the goal is to transfer Alice's part of the state to Bob by entanglement-assisted LOCC \cite{HOW05, HOW07}. The minimum amount of entanglement that must be consumed in extended quantum state merging (a variant of quantum state merging) is an operational interpretation of quantum discord \cite{CAB+11}, and thus NP-complete to compute.

Quantum discord quantifies the effect of decoherence in a family of protocols. It is the minimum difference between the yield of the fully quantum Slepian-Wolf (FQSW) protocol \cite{ADHW09} in the presence and absence of decoherence \cite{MD13}. The same holds for all descendant protocols of FQSW, where `yield' refers to the amount of entanglement consumed in quantum state merging \cite{MD11}, the amount of distilled entanglement in entanglement distillation \cite{SKB11}, the amount of classical information encoded in superdense coding, and the number of teleported qubits in quantum teleportation (see \cite{MD13} for details). Thus computing the minimum loss due to decoherence of the yield of all aforementioned protocols is NP-complete.

\section{Computational complexity of classical states}

A bipartite state $\rho_{AB}$ is separable if it can be expressed as
\begin{equation}
\rho_{AB}=\sum_ip_i|\psi_i^A\rangle\langle\psi_i^A|\otimes|\psi_i^B\rangle\langle\psi^B_i|,
\label{S}
\end{equation}
where $|\psi_i^A\rangle,|\psi_i^B\rangle$ are pure states in the subsystems $A,B$, respectively, and $p_i\ge0$ satisfies $\sum_ip_i=1$. $\rho_{AB}$ is quantum-classical if
\begin{equation}
\rho_{AB}=\sum_ip_i\rho_i^A\otimes\Pi_i^B,
\label{QC}
\end{equation}
where $\rho_i^A$'s are normalized, possibly mixed states in $A$, and $\{\Pi_i^B\}$ is a von Neumann measurement on $B$. $\rho_{AB}$ is quantum-classical if and only if $D(\rho_{AB}|B)=0$ \cite{OZ01} (note that $D_N(\rho_{AB})=0$ if and only if $D_P(\rho_{AB})=0$). $\rho_{AB}$ is classical-classical if
\begin{equation}
\rho_{AB}=\sum_{i,j}p_{ij}\Pi_i^A\otimes\Pi_j^B,
\end{equation}
where $p_{ij}\ge0$ satisfies $\sum_{i,j}p_{ij}=1$.

\begin{lemma}[NP-completeness of linear optimization over separable states \cite{Ioa07}]
Given an operator $O$ on a bipartite Hilbert space of dimension $m\times n$ and a real number $d$ with the promise that either (Y) $\max_{\rho_{AB}\in\cal{S}}\mathrm{tr}(\rho_{AB}O)\ge d$ or (N) $\max_{\rho_{AB}\in\cal{S}}\mathrm{tr}(\rho_{AB}O)\le d-\varepsilon$, it is NP-complete to decide which is the case, where $\varepsilon=1/\mathrm{poly}(m,n)$.
\end{lemma}

Let $\cal{QC}$ ($\cal{CC}$) be the set of all quantum-classical (classical-classical) states.

\begin{proposition}[NP-completeness of linear optimization over classical states]
Given an operator $O$ on a bipartite Hilbert space of dimension $m\times n$ and a real number $d$ with the promise that either (Y) $\max_{\rho_{AB}\in\cal{CC}}\mathrm{tr}(\rho_{AB}O)\ge d$ or (N) $\max_{\rho_{AB}\in\cal{CC}}\mathrm{tr}(\rho_{AB}O)\le d-\varepsilon$, it is NP-complete to decide which is the case, where $\varepsilon=1/\mathrm{poly}(m,n)$. The same holds for linear optimization over $\cal{QC}$.
\label{prop1}
\end{proposition}

\begin{proof}
Linear optimization over $\cal{CC}$ is in NP: the certificate of (Y) is the optimal state $\rho_{AB}$. $\cal{CC}\subseteq\cal{QC}\subseteq\cal{S}$ implies
\begin{equation}
\max_{\rho_{AB}\in\cal{CC}}\mathrm{tr}(\rho_{AB}O)\le\max_{\rho_{AB}\in\cal{QC}}\mathrm{tr}(\rho_{AB}O)\le\max_{\rho_{AB}\in\cal{S}}\mathrm{tr}(\rho_{AB}O).
\end{equation}
For any separable state $\sigma_{AB}=\sum_ip_i|\psi_i^A\rangle\langle\psi_i^A|\otimes|\psi_i^B\rangle\langle\psi^B_i|$,
\begin{equation}
\mathrm{tr}(\sigma_{AB}O)=\sum_ip_i\mathrm{tr}(|\psi_i^A\rangle\langle\psi_i^A|\otimes|\psi_i^B\rangle\langle\psi^B_i|O)\le\sum_ip_i\max_{\rho_{AB}\in\cal{CC}}\mathrm{tr}(\rho_{AB}O)=\max_{\rho_{AB}\in\cal{CC}}\mathrm{tr}(\rho_{AB}O),
\end{equation}
as $|\psi_i^A\rangle\langle\psi_i^A|\otimes|\psi_i^B\rangle\langle\psi^B_i|\in\cal{CC}$ and $\sum_ip_i=1$ \cite{RS10}. Thus,
\begin{equation}
\max_{\rho_{AB}\in\cal{CC}}\mathrm{tr}(\rho_{AB}O)=\max_{\rho_{AB}\in\cal{QC}}\mathrm{tr}(\rho_{AB}O)=\max_{\rho_{AB}\in\cal{S}}\mathrm{tr}(\rho_{AB}O).
\end{equation}
\end{proof}

\begin{lemma}[\cite{DR08, LL08}]
A bipartite quantum state $\rho_{AB}$ is separable if and only if there exists a state $\rho_{AA'|BB'}\in\cal{CC}$ in an extended Hilbert space such that $\rho_{AB}=\mathrm{tr}_{A'B'}\rho_{AA'BB'}$, or if and only if a state $\rho_{A|BB'}\in\cal{QC}$ exists such that $\rho_{AB}=\mathrm{tr}_{B'}\rho_{ABB'}$.
\end{lemma}

\begin{remark}
The definition (\ref{S}) of separability remains the same if the number of terms in the summation is restricted to be less than or equal to $m^2n^2$ \cite{Hor97}, where $m\times n$ is the dimension of the bipartite state $\rho_{AB}$. By slightly modifying the original proofs in \cite{DR08, LL08}, the dimensions of $\rho_{AA'|BB'}$ and $\rho_{A|BB'}$ can be required to be $m^3n^2\times m^2n^3$ and $m\times m^2n^3$, respectively.
\end{remark}

\begin{proposition}[NP-completeness of detecting classical states in a convex set]
Given a convex set $K$ of bipartite quantum states ($K$ is given by a polynomial-time algorithm outputting whether a state is in $K$) with the promise that either (Y) $K\cap\cal{CC}\neq\emptyset$ or (N) $\inf_{\rho\in K,\sigma\in\cal{CC}}\|\rho-\sigma\|_1\ge\delta$, it is NP-complete to decide which is the case, where $\delta=1/\mathrm{poly}(m,n)$. The same holds for detecting quantum-classical states in $K$.
\label{prop2}
\end{proposition}

\begin{proof}
Detecting classical-classical states in $K$ is in NP: the certificate of (Y) is an element in $K\cap\cal{CC}\neq\emptyset$. The hardness proof is a polynomial-time reduction from lemma \ref{Sep}. Given a bipartite state $\rho_{AB}$, define the convex set
\begin{equation}
K=\{\rho_{AA'|BB'}|\rho_{AB}=\mathrm{tr}_{A'B'}\rho_{AA'BB'}\}.
\end{equation}
(Y) If $\rho_{AB}$ is separable, then $K\cap\cal{CC}\neq\emptyset$. (N) If $\inf_{\sigma_{AB}\in\cal{S}}\|\rho_{AB}-\sigma_{AB}\|_2\ge\delta$, then for any $\rho_{AA'BB'}\in K$ and $\sigma_{AA'BB'}\in\cal{CC}$,
\begin{equation}
\|\rho_{AA'BB'}-\sigma_{AA'BB'}\|_1\ge\|\mathrm{tr}_{A'B'}(\rho_{AA'BB'}-\sigma_{AA'BB'})\|_1=\|\rho_{AB}-\sigma_{AB}\|_1\ge\|\rho_{AB}-\sigma_{AB}\|_2\ge\delta,
\end{equation}
as $\|\cdot\|_1$ is non-increasing under partial trace \cite{LZK08} and $\sigma_{AB}=\mathrm{tr}_{A'B'}\sigma_{AA'BB'}$ is separable. The NP-completeness of detecting quantum-classical states in $K$ can be proved analogously.
\end{proof}

\section{Conclusion and outlook}

We have proved that computing quantum discord is NP-complete. Therefore, the running time of any algorithm for computing quantum discord is believed to grow exponentially with the dimension of the Hilbert space so that computing quantum discord in a quantum system of moderate size is not possible in practice. As by-products, some entanglement measures and constrained Holevo capacity are NP-hard/NP-complete to compute. These complexity-theoretic results are directly applicable in quantum information processing, and may offer significant insights. Moreover, we have proved the NP-completeness of two typical problems related to classical states, providing evidence that working with classical states is generically computationally intractable.

The NP-completeness of computing quantum discord raises some interesting open problems. Is there an efficient approximation algorithm for computing quantum discord up to a moderate (e.g. constant additive) error? Can quantum discord be efficiently computed for certain important classes of states? What is the computational complexity of other measures of quantum correlation beyond entanglement (e.g. geometric quantum discord \cite{DVB10, LF10}, quantum deficit)? The computational complexity of quantum correlation in continuous-variable systems is a new research direction. In particular, Gaussian states are of great theoretical and experimental interest \cite{WHTH07, WPG+12}. The separability problem for multimode bipartite Gaussian states \cite{WW01} can be formulated as a semidefinite program \cite{HE06} and solved efficiently in theory and practice \cite{VB96} (the analog of lemma \ref{Sep} for Gaussian states is false). What is the computational complexity of Gaussian entanglement of formation \cite{WGK+04} and Gaussian quantum discord \cite{AD10, GP10}?

\section*{Acknowledgements}

The author would like to thank Sevag Gharibian, Joel E Moore and Mark M Wilde for useful comments. This work was supported by the ARO via the DARPA OLE program.

\end{document}